\documentclass[11pt]{article}
\usepackage[top=50mm, bottom=50mm, left=50mm, right=50mm]{geometry}
\usepackage{tocbibind}
\usepackage{amssymb}
\usepackage{amsmath}
\usepackage{amsthm}
\usepackage{epsfig}
\usepackage{graphicx}
\usepackage{graphics}
\usepackage{float}
\usepackage{subfigure}
\usepackage{multirow}
\usepackage{color}
\usepackage{lineno}
\usepackage{fullpage}
\usepackage[normalem]{ulem} 
\usepackage{makeidx}
\usepackage{xspace}
\usepackage{wrapfig}
\makeindex

\newtheorem{theorem}{Theorem}

\newtheorem{definition}{Definition}

\newtheorem{Exa}{Example}
\newtheorem{model}{Model}

\newcommand{\ntett}		{N^{\triangle_1}}
\newcommand{\ntetttre}		{N^{\triangle_1}_3}
\newcommand{\ntetttow}		{N^{\triangle_1}_2}

\newcommand{\keywords}
{
  \small	
  \textbf{\textit{Keywords: }} 
}

\begin{document}
\title{Relative Clustering Coefficient}
\author{Elena Farahbakhsh Touli, Oscar Lindberg }

\maketitle
\setcounter{page}{0}

\begin{abstract}
In this paper, we relatively extend the definition of global clustering coefficient to another clustering, which we call it \emph{relative clustering coefficient}. The idea of this definition is to ignore the edges in the network that the probability of having an edge is $0$. Here, we also consider a model as an example that using relative clustering coefficient is better than global clustering coefficient for comparing networks and also checking the properties of the networks. 
\end{abstract}


\newpage
\section{Introduction}\label{intro}
Recently in the field of physics and statistics, networks and graphs are two interesting topics for example Internet, email, media, and social network, citation network, and so on  \cite{1,2,3,Pohcn, Newman2001Random}. One of the properties of graphs is clustering and one of the most important characteristics of networks is that they are highly clustered. It is easy to see that the probability that a person in Germany and a person in Iran make a friendship is so low, but the probability that in a small city in Iran two friends of a person become friends is so high. And this is one of the important characteristics of networks in the real life. In the topological view of the graph, a highly clustered network contains a lot of triangles or cycles of length three.  \cite{Cslocn}

Networks or graphs \cite{Cslocn,4,GTwA} contain a set of vertices and a relation between the vertices. The relations between two vertices are defined by edges between the vertices and an edge between two vertices is shown by a line s.t. connects the two vertices. If two vertices have a relation, we add an edge between them otherwise, we do not add any edge between them. If the relationship is one-sided we have a directed graph otherwise, if the relationship is two-sided we have underacted graph. For example, friendship on Facebook is a two-sided relation, therefore if someone sends a request on Facebook to us we both become friends, and we can see what our friends share with us. But, friendship on Instagram is one-sided. When someone sends a friendship request to us, until we do not follow him or her, they will not be our friends. Edges in directed graphs are shown by using a line with an arrow indicating in which direction we have the relationship \cite{Cslocn,4,GTwA}. Examples of directed and undirected graphs are shown in Figure\ref{fig:dirnotdir}.

 \begin{figure}[htbp]
\begin{center}
\begin{tabular}{ccc}
\includegraphics[width=3cm]{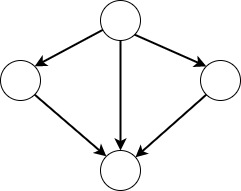} & \hspace*{0.3in}&
\includegraphics[width=3cm]{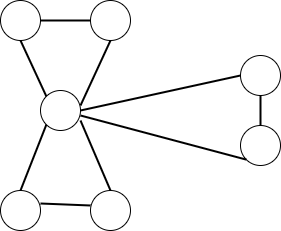} \\
(a) & & (b)
\end{tabular}
\end{center}
\vspace*{-0.25in}
\caption{(a) A directed graph.
(b) An undirected graph 
\label{fig:dirnotdir}}
\end{figure}

The clustering coefficient was defined by D.J. Watts, and S. H. Strogatz in 1998 \cite{Cdoswn} and it is used as a measure for comparing networks. We have two definitions for finding the clustering in the graph that we will talk about them in Section \ref{CC}; global clustering coefficient and local clustering coefficient. Using the definition described in \cite{Cdoswn}, we may define the local clustering coefficient as follows. For an individual $v$, we may consider the neighborhood of $v$ as the sub-graph containing only the neighbors of $v$. If this individual is connected to $k_v$ people, then there can exist at most  $k_v(k_v-1)/2$ friendships in the neighborhood. The local clustering coefficient for $v$ is the number of edges that exist in the neighborhood, divided by $k_v(k_v-1)/2$, i.e. the proportion of friendships that exist. However, the global clustering coefficient is calculated by the number of triangles divided by the number of triples that could make a triangle. But, after looking and thinking about some networks we see that these clustering coefficients are not sufficient for comparing all the networks and we need to establish and extend the clustering coefficients that we already have read about. And this was the idea of writing this paper.

 \noindent{\bf New work.} 
 In this paper, we extend the definition of clustering coefficient to a clustering that can indicate characteristics of networks better and we call it Relative Clustering Coefficient. At this definition, we consider only the edges in the network that we are allowed to add to the network. For example, if two people are in two different prisons they are not connected, although they may have the same lawyer. And if two people are in two different hospitals they are not connected, even though they may have the same doctor who works in those two hospitals.
 
 In this paper, we consider just simple undirected graphs. We will use and talk about some properties of graphs that we illustrate more here. They are bipartite graphs and cliques.
 \begin{figure}[htbp]
\begin{center}
\begin{tabular}{ccc}
\includegraphics[width=5cm]{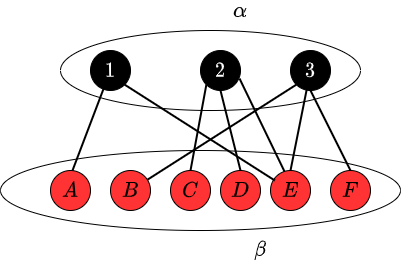} & \hspace*{0.3in}&
\includegraphics[width=4cm]{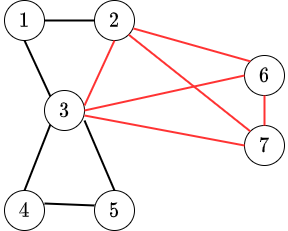} \\
(a) & & (b)
\end{tabular}
\end{center}
\vspace*{-0.25in}
\caption{(a) A bipartite graph with two parts $\alpha$ and $\beta$.
(b) A clique with vertices $2, 3, 6, 7$ and  red edges in the graph.
\label{fig:bicl}}
\end{figure} 
 
 {\bf{Bipartite Graphs}} \cite{GTwA}:In graph theory, a bipartite graph is a graph that we can divide vertices into two groups of vertices, such that there is no edge between the vertices in each group. (For more illustration look at Figure \ref{fig:bicl} part ($a$))
 
 {\bf Cliques} \cite{GTwA}: A complete graph is a simple graph (undirected graph without loops nor multiple edges) that includes all the possible edges between vertices or in other words, there is an edge between any two vertices of the graph. A subset of a graph is a graph that its vertices are the subset of the vertices of the main graph and edges are the subset of the edges of the main graph that connects vertices of the sub-graph. A clique is a sub-graph of a graph that is complete. (For more illustration look at Figure \ref{fig:bicl} part ($b$))

 Here the outline of this paper is as follows: in Section \ref{CC} we present the definition of clustering coefficients that we already had; global clustering coefficient and local clustering coefficient. We extend the definition of clustering coefficient at Section \ref{RCC}. We also present an example of a model and we use relative clustering coefficient instead of clustering coefficient. Last section which is Section \ref{C} is conclusion.

 \section{Clustering Coefficient}\label{CC}
 If we consider different networks in the real life we see that most of them are highly clustered, i.e. we can see that a friend of a friend of a person is a friend of the person as well. In another word, two friends of a person are with high probability friends in a way. From the topological view, we can see that there are lots of triangles in a network \cite{Pohcn, Cdoswn}.
 There are two definitions to measure the clustering in the network; local clustering coefficient and global clustering coefficient. 
 Here we want to illustrate these two definitions.
 
 \subsection{Local Clustering Coefficient}
 As it is defined in \cite{Cdoswn}, the local clustering coefficient for a vertex $v$ in a graph $G$ is defined as the number of triangles in the graph such that one vertex of the triangle is $v$ divided by the number of paths $\pi$ in $G$ with the length of $2$ such that the vertex $v$ is the middle vertex of $\pi$.

 In other words, for an undirected graph, if we consider $\nu_i$ as a set of vertices in the neighborhood of a vertex $v_i$, that means the set of vertices that there is an edge between $v_i$ and every vertex in the set, so we have that 
 $$\nu_i = \{v_i : e_{ij}\in E\}.$$
 
 Therefore, for a vertex $v_i$ we can define the local clustering coefficient $C_i$ as the number of edges between vertices in the set of $\nu_i$ divided by the number of edges that can exist between the vertices in the set $\nu_i$.

 Therefore, we can define the local clustering coefficient for a vertex $v_i$ as follows:

  \begin{definition}\cite{Pohcn} Local Clustering Coefficient
  
  For a vertex $v_i$, if $\nu_i$ is the set of vertices in the neighbourhood of $v_i$ and  $|\nu_i|$ is the size of this set, the local clustering $C_i$ for the vertex $v_i$ is defined as follows:
  
  \[
  C_i = \frac{|\{e_{jk}: v_j,v_k\in \nu_i ~\&~ e_{jk} \in E\}|}{\binom {|\nu_i|}{2}}
  \]
 \end{definition}
 
 We can define the local clustering coefficient for a directed graph in a similar way, but it is beyond the scope of this paper.
 
 \subsection{Global Clustering Coefficient}
 If we consider $N_\triangle$ as the number of triangles in the graph and $N_3$ as the number of sub-graphs containing three vertices that are connected at least by two edges (It means there is two edges between them or three edges, look at the Figure \ref{fig:cc}), the definition of clustering in a graph or network is as follows:
 \begin{figure}[htbp]
\begin{center}
\begin{tabular}{ccccc}
\includegraphics[width=3cm]{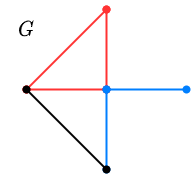} & \hspace*{0.3in}&
\includegraphics[width=1.5cm]{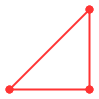} & \hspace*{0.35in}&
\includegraphics[width=1.5cm]{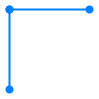} \\
(a) & & (b) & &  (c)
\end{tabular}
\end{center}
\vspace*{-0.25in}
\caption{(a) Graph $G$ and two of its sub-graphs with three vertices.
(b) The sub-graph with the color red is connected by three edges. (c) The sub-graph with the color blue is connected by two edges.
\label{fig:cc}}
\end{figure}

 \begin{definition}\cite{Pohcn} Global Clustering Coefficient
 
By using the notations $N_{\triangle}$ and $N_3$ that we defined earlier, the global Clustering Coefficient is defined as follows:

 \begin{equation*}
     C= \frac{3N_{\triangle}}{N_3}
 \end{equation*}
 \end{definition}
 
 After looking and considering other examples of networks we see that these definitions are not sufficient for considering high clustering in the network, so we define relative clustering coefficient. We consider the next section to illustrate and talk about the relative clustering coefficient.  Later we give an example that we see that considering the global clustering coefficient is not good enough.

 \section{Relative Clustering Coefficient (RCC)} \label{RCC}
 
 The idea of the definition of RCC is as follows that we just consider pairs of vertices that we can have an edge between them in the network. In other words, the probability of having an edge between them is larger than $0$.
 Here, for each pair of vertices, we define a capacity (a number, $0$ or $1$).  If we can have an edge between two vertices (Or if the probability of having an edge between two vertices is larger than $0$) the capacity is $1$, otherwise, the capacity is $0$.
 $N^{\triangle_1}$ is the number of all the triangles in the network that all the edges in the triangles have a capacity of  $1$, and  $N_3^{\triangle_1}$ is the number of triangles that the capacity of all the edges in it is $1$, such that all the edges in the triangle are in the network, and $N_2^{\triangle_1}$ is the number of triangles that the capacity of all edges in it is $1$, but just two of the edges in the triangle are in the network. Now, we define the relative clustering coefficient in a network as follows:
 
 \begin{definition}
 Relative Clustering Coefficient

By using the notation that we have illustrated earlier, we defined RCC as follows: $$C_R =\frac{ 3 N_3^{\triangle_1} }{ 3 N_3^{\triangle_1} +  N_2^{\triangle_1} }.$$
 \end{definition}

\begin{Exa}
If we consider a group of people ($A$ and $B$) in hospital number $1$ and two other people ($C$ and $D$) that are hospitalized in hospital number $2$, and the person $E$ a doctor that works with patients in both hospitals (For more information look at Figure \ref{fig:exampleccdif}) we have RCC and CC as follows:

$$C_R = 1,$$ while 
$$C<1.$$
But, in this case, we cannot add any edge to this network, because these two groups of people are separated and there is no physical contact between them. Therefore, the clustering coefficient should be $1$.

\begin{figure}[htbp]
\begin{center}
\begin{tabular}{ccc}
\includegraphics[width=4.5cm]{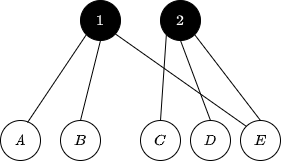} & \hspace*{0.2in}&
\includegraphics[width=4.5cm]{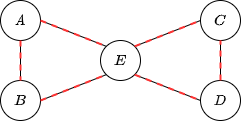} \\
(a) & & (b) 
\end{tabular}
\end{center}
\vspace*{-0.25in}
\caption{(a) A bipartite graph indicating that person $A$ and $B$ are hospitalized in hospital number $1$ and $C$ and  $D$ are hospitalized in hospital number $2$, and $E$ is working in both hospitals. 
(b) Full graph of the model, that contains all the possible edges. Edges with the capacity of $1$ are black. Dashed red lines indicate the edges in the network.
\label{fig:exampleccdif}}
\end{figure}
\end{Exa}

\subsection{RCC instead of CC}
In \cite{Pohcn} M. E. J. Newman illustrated a model for highly clustered networks. The model is as follows:

\begin{model}\label{model:1}\cite{Pohcn}

$\rightarrow{}$ We have $N$ individuals in total

$\rightarrow{}$ These individuals are divided into $M$ different groups.

$\rightarrow{}$
Individuals can belong to more than one group 

$\rightarrow{}$ 
Individuals belong to groups randomly

$\rightarrow{}$ If two individuals belong to one group with the probability of $p$ they are connected
otherwise they are not connected. 
\end{model}

To illustrate this model better, we use an example:

\begin{Exa}
We have some professors $($A$, $B$, $C$, $D$, $E$, $F$, $G$, $H$, $I$)$ who work at the University of Tehran. At the University of Tehran, we have some different departments ($1$, $2$, $3$, $4$, $5$) that each professor belongs to. Some professors work in different departments. Therefore, in meetings that are held in different departments that they participate in, they can meet other professors in the department. But, two professors who do not work in the same department, do not have any information from each other.  (For more information look at Figure \ref{fig:Example})

\begin{definition} Full Graph

If we construct a network that contains all the edges with the probability of $p>0$, we have a graph that we call it full graph. (See the lower figure in Figure \ref{fig:Example})
\end{definition}

\begin{figure}[htbp]
\begin{center}
\includegraphics[width=7cm]{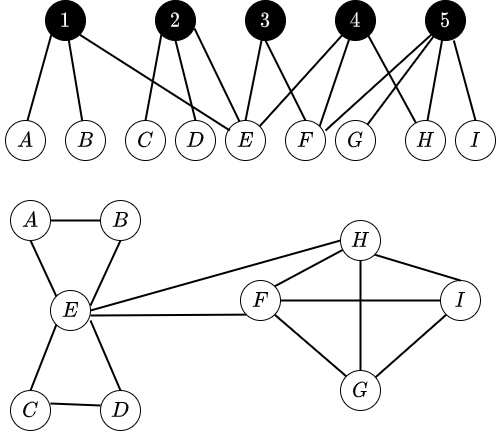} 
\end{center}
\vspace*{-0.25in}
\caption{The upper figure is a bipartite graph that indicates which professor belongs to which departments. Numbers indicate the departments and letters indicate the name of the professors. 
The lower figure is the full graph using this model. That is for example professor $E$ works for departments $1$, $2$, $3$, and $4$, therefore $E$ can know each people who work at those departments. But, there is no edge connecting $E$ to people who work at the department of $5$ that do not work in any of 
$1$ nor $2$ nor $3$ nor $4$.
\label{fig:Example}}
\end{figure}
\end{Exa}

For this model, M. E. J. Newman used the clustering coefficient and showed that 

\begin{equation*}
    C=pC'
\end{equation*}
such that $C'$ is the clustering of the full graph of the network. 
\\
\\
\noindent{\bf Reason for using RCC.~} 
Here, if we look at the full graph, we cannot add any edge to the network, so it is better if we use RCC instead of CC. Because it (the full graph) has all the edges between any pair of vertices and we are not allowed to add any other edges to the full graph.
Therefore, by using RCC we have the following theorem which is more reasonable to use for this example and model.

\begin{theorem}
  For the Model \ref{model:1} for large $n$, where $n$ is the number of vertex in the network,  we have that 
  $$C_R = p,$$ which $p$ is the probability of having an edge between two vertices in the network.
\end{theorem}
\begin{proof}
For each of the $\ntett$ cliques of size three (that also have capacity $1$ for each edge), we define $E_i$ as the number of edges within clique number $i = 1,2, ..., \ntett$. The number of given one of this cliques in binomially distributed with three numbers of trials and probability of success p, therefore $E_i \sim bin(p,3), i = 1,..., \ntett$.

The probability of all three edges existing within one of these cliques is 

\begin{equation*}
    P(E_i=3) = p^3
\end{equation*}
while the probability of exactly two edges existing with one of these cliques is
\begin{equation*}
    P(E_i=2) = 3p^2(1-p).
\end{equation*}

We define indicator variables as follows: 
\begin{equation*}
I(E_i=3) =
    \begin{cases}
      1, & \text{if}\ E_i=3 \\
      0, & \text{otherwise}
    \end{cases}
\end{equation*}

\[
I(E_i=2) =
    \begin{cases}
      1, & \text{if}\ E_i=2 \\
      0, & \text{otherwise}.
    \end{cases}
\]

Therefore, these indicator variables will have Bernolli distributions and expectations, so we have:
\begin{align*}
I(E_i =3 )&\sim be(p^3) \text{ has expectation } E[I(E_i = 3)] = p^3\\
I(E_i =2)&\sim be(3p^2(1-p)) \text{ has expectation } E[I(E_i = 2)] = 3p^2(1-p).
\end{align*}

Asymptotically we have that for large enough $n$ 
\begin{align*}
\ntetttre &= \sum_{i=1}^{\ntett} I(E_i = 3) = \theta(\ntett.E[I(E_i = 3)]) = \theta(\ntett.p^3)\\
\ntetttow &= \sum_{i=1}^{\ntett}I(E_i = 2) = \theta(\ntett.E[I(E_i = 2)]) = \theta(\ntett.3p^2(1-p))
\end{align*}
where the notation $\theta$ is defined in Definition \ref{def:theta}.

Therefore, for large enough $n$ the following relative clustering coefficient will asymptotcally go towards
\[
C_R = \frac{3\ntetttre}{3\ntetttre + \ntetttow} = \frac{3p^3\ntett}{3p^3\ntett + 3p^2(1-p)\ntett} = \frac{p}{p+(1-p)} = p.
\]
\end{proof}
 We can use a similar definition for defining relative local clustering coefficient as well. 

\begin{definition}\label{def:theta}\cite{ItA}
For two functions $f(n)$ and $g(n)$ we say that $f(n)=\theta g(n)$, if there exist two constant numbers $c_1$ and $c_2$ and an integer numbers $n'$ such that for all $n>n'$, we can write:
\[
c_1 g(n) \leq f(n)\leq c_2 g(n). 
\]
\end{definition}

\section{Conclusion}\label{C}
In this paper, we extend the definition of clustering coefficient to another definition which we called it relative clustering coefficient. This coefficient can measure the properties of networks better. In section \ref{CC} we defined two clustering coefficients, one local clustering, and another global clustering coefficient. We extended the definition of clustering coefficient in section \ref{RCC} and we used the clustering for measuring the property of a model as an example.

\newpage
\bibliography{bibil}

\begin{thebibliography}{10}

\bibitem{2}
R\'eka Albert and Albert-L\'aszl\'o Barab\'asi.
\newblock Statistical mechanics of complex networks.
\newblock {\em Reviews of Modern Physics}, 74(1):47--97, January 2002.

\bibitem{4}
Béla Bollobás.
\newblock {\em Modern Graph Theory}.
\newblock Graduate Texts in Mathematics 184. Springer-Verlag New York, 1
  edition, 1998.

\bibitem{GTwA}
J.~A. Bondy and U.~S.~R. Murty.
\newblock {\em Graph Theory with Applications}.
\newblock Elsevier, New York, 1976.

\bibitem{Cslocn}
{Caldarelli, G.}, {Pastor-Satorras, R.}, and {Vespignani, A.}
\newblock Structure of cycles and local ordering in complex networks.
\newblock {\em Eur. Phys. J. B}, 38(2):183--186, 2004.

\bibitem{ItA}
Thomas~H. Cormen, Charles~E. Leiserson, Ronald~L. Rivest, and Clifford Stein.
\newblock {\em Introduction to Algorithms}.
\newblock The MIT Press, 2nd edition, 2001.

\bibitem{3}
S.~N. Dorogovtsev and J.~F.~F. Mendes.
\newblock Evolution of networks with aging of sites.
\newblock {\em Physical Review E}, 62(2):1842–1845, Aug 2000.

\bibitem{Pohcn}
M.~E.~J. Newman.
\newblock Properties of highly clustered networks.
\newblock {\em Phys. Rev. E}, 68:026121, Aug 2003.

\bibitem{Newman2001Random}
M.~E.~J. Newman, S.~H. Strogatz, and D.~J. Watts.
\newblock Random graphs with arbitrary degree distributions and their
  applications.
\newblock {\em Phys. Rev. E}, 64(2):026118, July 2001.

\bibitem{1}
Steven~H. Strogatz.
\newblock {Exploring complex networks}.
\newblock {\em Nature}, 410(6825):268--276, March 2001.

\bibitem{Cdoswn}
Duncan~J. Watts and Steven~H. Strogatz.
\newblock Collective dynamics of ‘small-world’ networks.
\newblock {\em Nature}, 393(6684):440--442, 1998.

\end{thebibliography}

\end{document}